\documentclass[a4paper, UKenglish, cleveref, autoref, thm-restate, numberwithinsect]{lipics-v2019}

\usepackage[ruled]{algorithm}
\usepackage{mleftright}
\usepackage[noadjust]{cite}
\usepackage{amsthm}
\newtheorem{observation}[theorem]{Observation}
\DeclareMathOperator{\polylog}{polylog}

\newcommand{\To}{\rightsquigarrow}

\newcommand{\caD}{\mathcal{D}}
\newcommand{\caP}{\mathcal{P}}

\newcommand{\MM}{\mathsf{MM}}
\newcommand{\Fast}{\mathrm{Fast}}
\newcommand{\Extend}{\mathrm{Extend}}
\newcommand{\Tin}{T^{\sf in}}
\newcommand{\Tout}{T^{\sf out}}
\newcommand{\CL}{\mathsf{CL}}
\newcommand{\CR}{\mathsf{CR}}
\newcommand{\BCP}{\mathsf{BCP}}
\newcommand{\parent}{\mathrm{parent}}

\title{Improved Distance Sensitivity Oracles with Subcubic Preprocessing Time}

\titlerunning{Improved DSOs with Subcubic Preprocessing Time}

\author{Hanlin Ren}{Institute for Interdisciplinary Information Sciences, Tsinghua University, China}{h4n1in.r3n@gmail.com}{}{}{}

\authorrunning{H. Ren}

\Copyright{Hanlin Ren}



\acknowledgements{
	I would like to thank Zihao Li for introducing this problem to me, Ran Duan and Yong Gu for helpful discussions, and Ce Jin for helpful comments that improve the presentation of this paper. I would also like to thank Hongxun Wu for reading and commenting on an early draft version of this paper, and pointing out the issue with non-unique shortest paths. I am also grateful to the anonymous referees of ESA and JCSS for their helpful comments that improve the presentation of this paper. In particular, I thank an anonymous JCSS referee for suggesting I explain the path-reporting algorithm and an anonymous ESA referee for suggesting I consider the more general case that edge weights may be negative. (Unfortunately, I am not aware of any easy way to generalize our results to handle negative edge weights.)
}

\nolinenumbers

\hideLIPIcs

\begin{document}
	\maketitle
	
	\begin{abstract}
		We consider the problem of building \emph{distance sensitivity oracles} (DSOs). Given a directed graph $G=(V, E)$ with edge weights in $\{1, 2, \dots, M\}$, we need to preprocess it into a data structure, and answer the following queries: given vertices $u,v\in V$ and a failed vertex or edge $f\in (V\cup E)$, output the length of the shortest path from $u$ to $v$ that does not go through $f$. Our main result is a simple DSO with $\tilde{O}(n^{2.7233}M)$ preprocessing time and $O(1)$ query time. Moreover, if the input graph is undirected, the preprocessing time can be improved to $\tilde{O}(n^{2.6865}M)$. The preprocessing algorithm is randomized with correct probability $\ge 1-1/n^C$, for a constant $C$ that can be made arbitrarily large. Previously, there is a DSO with $\tilde{O}(n^{2.8729}M)$ preprocessing time and $\polylog(n)$ query time [Chechik and Cohen, STOC'20].
		
		At the core of our DSO is the following observation from [Bernstein and Karger, STOC'09]: if there is a DSO with preprocessing time $P$ and query time $Q$, then we can construct a DSO with preprocessing time $P+\tilde{O}(n^2)\cdot Q$ and query time $O(1)$. (Here $\tilde{O}(\cdot)$ hides $\polylog(n)$ factors.)
	\end{abstract}
	
	\section{Introduction}\label{sec:intro}
	
	Suppose we are given a directed graph $G=(V,E)$, and we want to build a data structure that, given any two vertices $u,v\in V$ and a failure $f$, which is either a failed vertex or a failed edge, outputs the length of the shortest path from $u$ to $v$ that does not go through $f$. Such a data structure is called a \emph{distance sensitivity oracle} (or DSO for short).
	
	The problem of constructing DSOs is motivated by the fact that real-life networks often suffer from failures. Suppose we have a network with $n$ nodes and $m$ (directed) links, and we want to send a package from a node $u$ to another node $v$. Normally, it suffices to compute the shortest path from $u$ to $v$. However, if some node or link $f$ in this network fails, then our path cannot use $f$, and our task becomes to find the shortest path from $u$ to $v$ that does not go through $f$. Usually, there is only a very small number of failures. In this paper, we consider the simplest case, in which there is only one failed node or link.
	
	The problem of constructing a DSO is well-studied: Demetrescu et al.~\cite{DemetrescuTCR08} showed that given a directed graph $G=(V, E)$, there is a DSO which occupies $O(n^2\log n)$ space and can answer a query in $O(1)$ time. Duan and Zhang \cite{DuanZ17a} improved the space complexity to $O(n^2)$, which is optimal for dense graphs (i.e.~$m=\Theta(n^2)$).
	
	Unfortunately, the oracle in \cite{DemetrescuTCR08} requires a large preprocessing time ($O(mn^2+n^3\log n)$). In real-life applications, the preprocessing time of the DSO is also very important. Bernstein and Karger \cite{BernsteinK08, BernsteinK09} improved this time bound to $\tilde{O}(mn)$. Note that the All-Pairs Shortest Paths (APSP) problem, which only asks the distances between each pair of vertices $u,v$, is conjectured to require $mn^{1-o(1)}$ time to solve \cite{LincolnWW18}. Since we can solve the APSP problem by using a DSO, the preprocessing time $\tilde{O}(mn)$ is optimal in this sense.
	
	However, if the edge weights are small positive integers (that do not exceed $M$), then the APSP problem can be solved in $\tilde{O}(n^{2.5286}M)$ time \cite{Zwick02}\footnote{\cite{Zwick02} only claims a time bound of $\tilde{O}(n^{2.58}M)$; the time bound $\tilde{O}(n^{2.5286}M)$ is calculated using improved time bounds for rectangular matrix multiplication, see \cite{GallU18}.}, which is significantly faster than $O(mn)$ for dense graphs with small weights (e.g.~$M=O(1)$). Thus it might be possible to obtain better results than \cite{BernsteinK09} in the regime of small integer edge weights. Weimann and Yuster \cite{WeimannY13} showed that for any constant $\alpha\in(0,1)$, we can construct a DSO in $\tilde{O}(n^{1-\alpha+\omega}M)$ time.  Here $\omega<2.3728639$ is the exponent of matrix multiplication \cite{Gall14a}. However, the query time for this oracle is $\tilde{O}(n^{1+\alpha})$, which is \emph{superlinear}. Later, Grandoni and Williams \cite{GrandoniW12} showed that for every constant $\alpha\in[0,1]$, we can construct a DSO in $\tilde{O}(n^{\omega+1/2}M+n^{\omega+\alpha(4-\omega)}M)$ time, which answers each query in $\tilde{O}(n^{1-\alpha})$ time.
	
	Recently, in an independent work, Chechik and Cohen \cite{ChechikC20} showed that a DSO with $\polylog(n)$ query time can be constructed in $\tilde{O}(Mn^{2.873})$ time, achieving \emph{both} subcubic preprocessing time and polylogarithmic query time. The space complexity for their DSO is $\tilde{O}(n^{2.5})$.
	
	\subsection{Our Results}
	In this work, we show improved and simplified constructions of DSOs. We start with an observation.
	
	\begin{observation}[Informal]\label{lemma:complexity-of-DSO}
		If we have a DSO with preprocessing time $P$ and query time $Q$, then we can build a DSO with preprocessing time $P+\tilde{O}(n^2)\cdot Q$ and query time $O(1)$.
	\end{observation}
	
	For $\alpha=0.2$, the oracle in \cite{GrandoniW12} already achieves $\tilde{O}(n^{2.8729}M)$ preprocessing time and $O(n^{0.8})$ query time. \cref{lemma:complexity-of-DSO} implies that this query time can be brought down to $O(1)$.
	
	\cref{lemma:complexity-of-DSO} can be proved by a close inspection of \cite{BernsteinK09}: The algorithm in \cite{BernsteinK09} for constructing a DSO picks $\tilde{O}(n^2)$ carefully chosen queries $(u,v,f)$, such that the answers of all these queries can be computed in $\tilde{O}(mn)$ time. Then, from these answers, we can easily compute a DSO with query time $O(1)$. If, instead of computing these answers in $\tilde{O}(mn)$ time, we use the given DSO to answer these queries, the preprocessing time becomes $P+\tilde{O}(n^2)\cdot Q$.
	
	Our main result is a simple construction of DSOs with preprocessing time $\tilde{O}(n^{2.7233}M)$ and query time $O(1)$. If the input graph is undirected, we can achieve a better preprocessing time of $\tilde{O}(n^{2.6865}M)$.
	
	\begin{restatable}{theorem}{ThmDSO}\label{thm:DSO}
		We can construct a DSO with $\tilde{O}(n^{2.7233}M)$ preprocessing time and $O(1)$ query time. Moreover, if the input graph is \emph{undirected}, then we can construct a DSO with $\tilde{O}(n^{(3+\omega)/2}M)=\tilde{O}(n^{2.6865}M)$ preprocessing time and $O(1)$ query time. The construction algorithms are randomized and yield valid DSOs w.h.p.\footnote{We say that an event happens \emph{with high probability} (w.h.p.), if it happens with probability $1-1/n^C$, for some constant $C$ that can be made arbitrarily large.}
	\end{restatable}
	
	Compared to previous constructions of DSOs, our results have a better dependence on $n$ and are conceptually simpler. Also, the size of our DSO is $\tilde{O}(n^2)$, which is better than \cite{ChechikC20}. However, one drawback of our construction is that it cannot handle negative edge weights. (The preprocessing and query time bounds stated above for \cite{WeimannY13, GrandoniW12, ChechikC20} still hold when the edge weights are integers in $[-M,M]$.)
	
	\subsection{Non-unique Shortest Paths}\label{sec:non-unique}
	
	A subtle technical issue is that the shortest paths in $G$ between some vertex pairs may not be unique. Normally, if we perturb every edge weight by a small random value, then we can ensure that shortest paths are unique w.h.p.~by the isolation lemma~\cite{MulmuleyVV87, Ta-Shma15}. However, the subcubic-time algorithms for APSP~\cite{Seidel95, ShoshanZ99, Zwick02} depend crucially on the assumption that edge weights are small integers, so we cannot perform the random perturbation.
	
	Fortunately, there is another elegant way of breaking ties, described in Section 3.4 of \cite{DemetrescuI04}, that is suitable for our propose. We arbitrarily assign a linear order $\prec$ over the edges of $G$. (For example, we index the edges by distinct numbers in $\{1,2,\dots,m\}$, and define $e_1\prec e_2$ if the index of $e_1$ is smaller than the index of $e_2$.) For a path $p$, define the \emph{bottleneck edge} of $p$ as the smallest edge (w.r.t.~$\prec$) in $p$. Let $G'$ be a subgraph of $G$ (think of $G'=G$ or $G'$ is $G$ with some vertex or edge removed), define $w_{G'}(u, v)$ as the largest (w.r.t.~$\prec$) edge $w\in G'$ such that there is a \emph{shortest} path in $G'$ from $u$ to $v$ whose bottleneck edge is $w$. Define $\rho_{G'}(u, v)$ to be the following path: If $u=v$, then $\rho_{G'}(u, v)$ is the empty path; otherwise, suppose $w_{G'}(u, v)$ is an edge from $u^\star$ to $v^\star$, then $\rho_{G'}(u, v)$ is the concatenation of $\rho_{G'}(u, u^\star)$, $w_{G'}(u, v)$ and $\rho_{G'}(v^\star, v)$. Note that $\rho_{G'}(u, v)$ is always well-defined: we can define $\rho_{G'}(u, v)$ inductively, in non-decreasing order of the distance from $u$ to $v$. By a similar induction, we can also prove that $\rho_{G'}(u, v)$ is always a shortest path from $u$ to $v$ in $G'$.

	We list a few properties of this tie-breaking method in the following theorem.
	
	\begin{restatable}{theorem}{thmProperties}\label{thm:property-APBSP}
		The following properties of $\rho_{G'}(u, v)$ are true:
		\begin{enumerate}[(Property a)]
			\item For every $u', v'\in \rho_{G'}(u, v)$ such that $u'$ appears before $v'$, the portion of $u'\To v'$ in $\rho_{G'}(u, v)$ coincides with the path $\rho_{G'}(u', v')$.\label{prop:consistency}
			\item Let $G'$ be a subgraph of $G$, suppose $\rho_G(u, v)$ is completely contained in $G'$, then $\rho_{G'}(u, v) = \rho_G(u, v)$.\label{prop:subgraph}
		\end{enumerate}
	\end{restatable}

	For every vertex $u\in V$, we can compute the collection of paths $\rho_{G'}(v, u)$ for every $v\in V$. By (Property \ref{prop:consistency}), these paths form a tree rooted at $u$, and we say this tree is the \emph{incoming shortest path tree} rooted at $u$, denoted as $\Tin(u)$. Similarly, we can define the outgoing shortest path trees rooted at each vertex, denoted as $\Tout(u)$. The advantage of this tie-breaking method is:
	
	\begin{restatable}[\cite{DuanP09}]{theorem}{thmAPBSP}\label{thm:APBSP}
		Given a subgraph $G'$, the incoming and outgoing shortest path trees $\Tin(u)$ and $\Tout(u)$ for each vertex $u$ can be computed in $\tilde{O}(n^{(3+\omega)/2}M^{1/2})$ time.
	\end{restatable}

	For completeness we will prove \cref{thm:property-APBSP} and \ref{thm:APBSP} in \cref{sec:APBSP}.
	
	\subsection{Notation}
	We mainly stick to the notation used in \cite{DuanP09a}, namely:\begin{itemize}
		\item If $p$ is a path, then $|p|$ denotes the number of edges in it, and $\|p\|$ denotes its length (i.e.~the total weight of its edges).
		\item Let $u,v\in V$, we denote $uv = \rho_G(u, v)$, i.e.~the unique shortest path from $u$ to $v$ in the original graph. For a vertex or edge failure $f$, we denote $uv\diamond f=\rho_{G-f}(u, v)$, where $G-f$ is the subgraph of $G$ obtained by removing the failure $f$. Note that if $f$ is not in the original path $uv$, then by (Property \ref{prop:subgraph}), $uv\diamond f=uv$.
		\item Let $p$ be a path from $u$ to $v$. For two vertices $a,b\in p$ such that $a$ appears earlier than $b$, we say the interval $p[a,b]$ is the subpath from $a$ to $b$, and $p(a,b)$ is the path $p[a,b]$ without its endpoints ($a$ and $b$). If the path $p$ is known in the context, then we may omit $p$ and simply write $[a,b]$ or $(a,b)$. (Property \ref{prop:consistency}) implies that, if $p$ is the path $uv$ (or $uv\diamond f$ respectively), then $p[a, b]$ is the path $ab$ (or $ab\diamond f$ respectively).
	\end{itemize}
	
	We define $\MM(n_1,n_2,n_3)$ as the complexity of multiplying an $n_1\times n_2$ matrix and an $n_2\times n_3$ matrix. Let $a,b,c$ be real numbers, we define $\omega(a,b,c)$ be the infimum over all real numbers $\alpha$ such that $\MM(n^a,n^b,n^c)=O(n^\alpha)$. For any real number $r$, we have $\omega(1,1,r)=\omega(1,r,1)=\omega(r,1,1)$ \cite{LottiR83}, and we denote $\omega(r)=\omega(1,1,r)$.
	
	We also need the following adaptation of Zwick's APSP algorithm \cite{Zwick02} (see also \cite[Corollary 3.1]{GrandoniW12}):
	\begin{theorem}\label{thm:truncated-Zwi02}
		Given an integer $r$ and a directed graph $G=(V,E)$ with edge weights in $\{1,2,\dots,M\}$, we can compute the distances between every pair of nodes whose shortest path uses at most $r$ edges, in $\tilde{O}(rM\cdot \MM(n,n/r,n))$ time.
	\end{theorem}
	\begin{proof}[Proof Sketch]
		We simply run the first $\lceil\log_{3/2}r\rceil$ iterations of the algorithm {\bf rand-short-path} in \cite{Zwick02}. The correctness of this algorithm is guaranteed by \cite[Lemma 4.2]{Zwick02}.
	\end{proof}
	
	\section{Constructing a DSO in $\tilde{O}(n^{2.7233}M)$ Time}\label{sec:some-DSO}
	
	In this section, we prove \cref{thm:DSO}.
	
	{
		\def\footnote#1{}
		\ThmDSO*
	}
	
	Given an integer $r$ and a graph $G=(V,E)$, we define an \emph{$r$-truncated DSO} to be a data structure that, when given a query $(u,v,f)$, outputs the value $\min\{\|uv\diamond f\|,r\}$. In other words, an $r$-truncated DSO is a DSO which only needs to be correct when the path $uv\diamond f$ has length at most $r$. If this length is greater than $r$, it outputs $r$ instead.
	
	Inspired by Zwick's APSP algorithm \cite{Zwick02}, our main idea is to compute an $r$-truncated DSO for every $r=(3/2)^i$. Our high-level strategies for small $r$ and large $r$ are completely different as described below.
	
	When $r$ is small, the sampling approach in \cite{WeimannY13, GrandoniW12} already suffices. Fix a particular query $(u,v,f)$, we assume that $f$ is a vertex failure, and $\|uv\diamond f\|\le r$. In particular, since the edge weights are positive, there are at most $r+1$ vertices in the path $uv\diamond f$. Suppose we sample a graph by discarding each vertex w.p.~$1/r$. With probability $\Omega(1/r)$, the resulting graph would ``capture'' this query in the sense that $f$ is not in it but $uv\diamond f$ is completely included in it. Therefore, if we take $\tilde{O}(r)$ independent samples, and compute APSP for each sampled subgraph, we can deal with every vertex-failure query w.h.p. If $f$ is an edge failure, we simply discard every \emph{edge} (instead of vertex) in every subgraph w.p.~$1/r$, and we can still deal with every edge-failure query w.h.p.
	
	For large $r$, our idea is to compute a $(3/2)r$-truncated DSO from an $r$-truncated DSO. More precisely, given an $r$-truncated DSO with $O(1)$ query time, we can compute a $(3/2)r$-truncated DSO with $\tilde{O}(Mn/r)$ query time as follows. First we sample a bridging set (see~\cite{Zwick02}) $H$ of size $\tilde{O}(Mn/r)$. Let $(u,v,f)$ be a query such that $r\le \|uv\diamond f\|\le (3/2)r$, then w.h.p.~there is a ``bridging vertex'' $h\in H$ such that $h$ is on the path $uv\diamond f$, and both of the queries $(u,h,f)$ and $(h,v,f)$ are captured by the $r$-truncated DSO. If we iterate through $H$, we can answer the query $(u,v,f)$ in $\tilde{O}(Mn/r)$ time. Then we use an ``$r$-truncated'' version of \cref{lemma:complexity-of-DSO} to transform this $(3/2)r$-truncated DSO with slow query time into a new one with $O(1)$ query time.
	
	Now, we present how to implement the above ideas in detail.
	
	\subsection{Case I: $r$ is Small}\label{sec:r-small}
	\def\e{\mathsf{e}}
	\def\v{\mathsf{v}}
	Let $\tilde{r}=\lceil 8Cr\ln n\rceil$, where $C$ is a large enough constant. We independently sample $2\tilde{r}$ graphs $G_1^\e,G_2^\e,\dots,G_{\tilde{r}}^\e$ and $G_1^\v,G_2^\v,\dots,G_{\tilde{r}}^\v$. (The superscripts $\e$ and $\v$ stand for ``edge'' and ``vertex'' respectively.) Each graph is sampled as follows:
	\begin{itemize}
		\item The vertex set of each $G_i^\e$ is equal to the original vertex set $V$, and the edge set of each $G_i^\e$ is sampled by including every edge independently w.p.~$1-1/r$.
		\item The vertex set of each $G_i^\v$ is sampled by including every vertex independently w.p.~$1-1/r$, and the graph $G_i^\v$ is the induced subgraph of $G$ on vertices $V(G_i^\v)$.
	\end{itemize}
	Then, for each $1\le i\le \tilde{r}$, we compute all-pairs shortest paths of the graph $G_i^\e$ and $G_i^\v$, but we only compute the shortest paths that use at most $r$ edges. By \cref{thm:truncated-Zwi02}, this step can be done in $\tilde{O}(rM\cdot \MM(n,n/r,n))$ time for each graph. Alternatively, if the input graph is undirected, then this step can be done in $\tilde{O}(Mn^\omega)$ time \cite{Seidel95, ShoshanZ99} for each graph.
	
	Consider a query $(u,v,f)$, suppose $f$ is a vertex failure, and assume that $\|uv\diamond f\|\le r$. Let $1\le i\le \tilde{r}$, we say $i$ is \emph{good} for the query $(u,v,f)$, if both of the following hold.\begin{itemize}
		\item The graph $G_i^\v$ does not contain $f$.
		\item The graph $G_i^\v$ contains the entire path $uv\diamond f$.
	\end{itemize}
	
	For every $i$ ($1\le i\le \tilde{r}$), the probability that $i$ is good for the particular query $(u,v,f)$ is at least
	\[\left(1/r\right)\cdot \left(1-1/r\right)^{r+1}\ge 1/8r\text{ (if $r\ge 2$).}\]
	
	Given a query $(u,v,f)$, where $f$ is a vertex failure, we iterate through every $i$ such that $f\not\in V(G_i^\v)$, and take the smallest value among the distances from $u$ to $v$ in these graphs $G_i^\v$. With high probability, there are only $\tilde{O}(1)$ valid indices $i$ such that $f\not\in V(G_i^\v)$, and we can preprocess this set of indices $i$ for every $f\in V$. Therefore the query time is $\tilde{O}(1)$.
	
	The algorithm succeeds at a query $(u,v,f)$ if there is an $i$ that is good for $(u,v,f)$. Since the graphs $G_i^\v$ are independent, the probability that there is an $i$ good for $(u,v,f)$ is at least
	\[1-\left(1-1/8r\right)^{\tilde{r}}\ge 1-1/n^C.\]
	By a union bound over all $n^3$ triples of possible queries $(u,v,f)$, it follows that our data structure is correct w.p.~at least $1-1/n^{C-3}$, which is a high probability.
	
	If $f$ is an edge failure, we look at the graphs $G_i^\e$ instead of $G_i^\v$. We say $i$ is good if the graph $G_i^\e$ does not contain $f$ but contains the entire path $uv\diamond f$. Again, if we fix $(u, v, f)$, then the probability that a particular $i$ is good is at least $\Omega(1/r)$, so w.h.p.~there is some $i$ that is good for $(u, v, f)$. By a union bound over all $O(n^4)$ triples $(u, v, f)$, our data structure is still correct w.h.p.
	
	In conclusion, there is an $r$-truncated DSO with $\tilde{O}(1)$ query time, whose preprocessing time is
	$\tilde{O}(\tilde{r}\cdot rM\cdot \MM(n,n/r,n))$
	for directed graphs, and
	$\tilde{O}(\tilde{r}\cdot Mn^\omega)$
	for undirected graphs.
	
	\subsection{An Observation}\label{sec:obs}
	We need the following observation (``$r$-truncated'' version of \cref{lemma:complexity-of-DSO}), which roughly states that given an $r$-truncated DSO with preprocessing time $P$ and query time $Q$, we can build an $r$-truncated DSO with preprocessing time $P+\tilde{O}(n^2)\cdot Q$ and query time $O(1)$. More formally, we have:
	
	\begin{restatable}{observation}{ObsTruncatedBK}		\label{ob:fast-r-truncated-DSO-formal}
		Let $r$ be an integer, $G=(V,E)$ be an input graph, and $\caD$ be an arbitrary $r$-truncated DSO. We can construct $\Fast(\caD)$, which is an $r$-truncated DSO with $O(1)$ query time and a preprocessing algorithm as follows.\begin{itemize}
			\item It takes as input a graph $G$, the distance matrix of $G$, and the (incoming and outgoing) shortest path trees rooted at each vertex.
			\item Then it invokes the preprocessing algorithm of $\caD$ on the input graph $G$.
			\item At last, it makes $\tilde{O}(n^2)$ queries to $\caD$, and spends $\tilde{O}(n^2)$ additional time to finish the preprocessing algorithm.
		\end{itemize}
	\end{restatable}
	
	We emphasize the following technical details that are not reflected in the informal statement of \cref{lemma:complexity-of-DSO}. First, besides the input graph $G$, we also need the distance matrix and shortest path trees of $G$ (henceforth the ``APSP data'' of $G$) before using $\caD$. We can compute the APSP data using \cref{thm:APBSP}. Second, the preprocessing algorithm of $\caD$ is called \emph{only once}, and on the same graph $G$ (we already have the APSP data of $G$). The reason that the second detail is important is: Suppose we have another routine that given an $r$-truncated DSO $\caD$, constructs $\Extend(\caD)$ which is a $(3/2)r$-truncated DSO with a possibly large query time. Then given a $1$-truncated DSO $\caD^{\sf start}$, we can construct a (normal) DSO as follows:
	\[\caD^{\sf final} = \underbrace{\Fast(\Extend(\Fast(\Extend(\dots \Extend(\caD^{\sf start})))))}_{O(\log n)\text{ times}}.\]
	
	However, even if the preprocessing algorithm of $\Fast(\caD)$ invokes the preprocessing algorithm of $\caD$ \emph{twice}, the preprocessing algorithm of $\caD^{\sf final}$ would invoke \emph{a polynomial times} the preprocessing algorithm of $\caD^{\sf start}$, which is too many. In contrast, if the preprocessing algorithms of both $\Fast(\caD)$ and $\Extend(\caD)$ only invoke the preprocessing algorithm of $\caD$ once, then the preprocessing algorithm of $\caD^{\sf final}$ would also invoke the preprocessing algorithm of $\caD^{\sf start}$ only once, which is okay.
	
	\subsection{Case II: $r$ is Large}\label{sec:r-large}
	Suppose we have an $r$-truncated DSO $\caD$, which has preprocessing time $P$ and query time $O(1)$. We show how to construct a $(3/2)r$-truncated DSO, which we name as $\Extend(\caD)$, with preprocessing time $P+O(n^2)$ and query time $\tilde{O}(nM/r)$. This is done by the following bridging set argument.
	
	\def\m{\mathsf{mid}}
	
	Let $\caP$ be the set of paths $uv\diamond f$, for every $u,v\in V$ and $f\in(V\cup E)$ such that $r\le \|uv\diamond f\|<(3/2)r$. This corresponds to the paths that $\caD$ cannot deal with, but $\Extend(\caD)$ has to output the correct answer. Let $p=uv\diamond f\in\caP$, $\m(p)$ be the set of vertices $y\in p$ such that $\|p[u,y]\|<r$ and $\|p[y,v]\|<r$. (See \cref{fig:mid}.) For every $y\in\m(p)$, since $p[u,y] = uy\diamond f$ and $p[y,v] = yv\diamond f$, it follows that $\caD$ can find $\|uy\diamond f\|$ and $\|yv\diamond f\|$ correctly. Moreover, $|\m(p)|\ge r/3M$.
	
	\begin{figure}[H]
		\centering
		\includegraphics[width=0.9\linewidth]{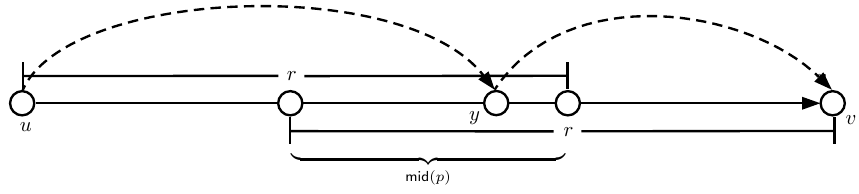}
		\caption{Example of a path $p=uv\diamond f$. If we can find a vertex $y\in \m(p)$, then we can use $\caD$ to compute $\|uy\diamond f\|$ and $\|yv\diamond f\|$, thus to compute the length of $p$.}\label{fig:mid}
	\end{figure}
	
	Fix a large enough constant $C$, the preprocessing algorithm of $\Extend(\caD)$ is as follows: We preprocess $\caD$, and then randomly sample a set $H$ of vertices, where every vertex $v\in V$ is in $H$ with probability $\min\{1,3CM\ln n/r\}$ independently. We have $|H|=\tilde{O}(nM/r)$ w.h.p.
	
	Fix $u,v\in V$ and $f\in (V\cup E)$, suppose $p=uv\diamond f$ and $r\le \|p\|<(3/2)r$. Then the probability that $H$ hits $\m(p)$ (i.e.~$H\cap \m(p)\ne\varnothing$) is at least
	\[1-(1-3CM\ln n/r)^{r/3M}\ge 1-1/n^C.\]
	
	By a union bound over $O(n^4)$ paths in $\caP$, it follows that w.h.p.~$H$ hits $\m(p)$ for every path $p\in\caP$.
	
	The query algorithm for $\Extend(\caD)$ is as follows: Given a query $(u,v,f)$, if $\caD(u,v,f)<r$, then we output $\caD(u,v,f)$; otherwise we output
	\[\min\left\{(3/2)r,\min_{h\in H, \caD(u,h,f)<r, \caD(h,v,f)<r}\left\{\caD(u,h,f)+\caD(h,v,f)\right\}\right\}.\]
	It is easy to see that $\Extend(\caD)$ is a correct $(3/2)r$-truncated DSO, has preprocessing time $P+O(n^2)$ and query time $\tilde{O}(nM/r)$.
	
	\subsection{Putting It Together}\label{sec:putting-it-together}
	Let $a\in[0,1]$ be a constant that we pick later, and $r=n^a$. We first compute the APSP data using \cref{thm:APBSP}, which costs $\tilde{O}(n^{(3+\omega)/2}M^{1/2})$ time and will not be the bottleneck. Then we invoke \cref{sec:r-small} to build an $r$-truncated DSO $\caD^0$ for $r=n^a$, which costs $\tilde{O}(r^2M\cdot \MM(n,n/r,n))$ time for directed graphs or $\tilde{O}(r\cdot Mn^\omega)$ for undirected graphs. Then for every $1\le i\le \lceil \log_{3/2}(Mn/r)\rceil$, suppose we have an $r(3/2)^{i-1}$-truncated DSO $\caD^{i-1}$, we can construct $\caD^i=\Fast(\Extend(\caD^{i-1}))$ which is an $r(3/2)^i$-truncated DSO. This step costs $\tilde{O}(n^3M/(r(3/2)^i))$ time. The preprocessing algorithm terminates when $i=i_\star=\lceil\log_{3/2}(Mn/r)\rceil=O(\log n)$, and we obtain an $r(3/2)^{i_\star}$-truncated DSO which is a (normal) DSO.
	
	\subparagraph{Case 1: the input graph is undirected.} The total preprocessing time is
	\[\tilde{O}\mleft(r\cdot Mn^\omega+n^3M/r\mright)=\tilde{O}\mleft(n^{\max\{\omega+a,3-a\}}M\mright).\]
	When $a=(3-\omega)/2$, this time complexity is $\tilde{O}(n^{(3+\omega)/2}M)=\tilde{O}(n^{2.6865}M)$.
	
	Therefore, given an undirected graph $G=(V,E)$, there is a DSO with $\tilde{O}(n^{2.6865}M)$ preprocessing time and $O(1)$ query time.
	
	\subparagraph{Case 2: the input graph is directed.} The total preprocessing time is
	\begin{equation*}
		\tilde{O}\mleft(r^2M\cdot \MM(n,n/r,n)+n^3M/r\mright)=\tilde{O}\mleft(n^{2a+\omega(1-a)}M+n^{3-a}M\mright).%
	\end{equation*}
	(Recall that $\omega(1-a)$ is the exponent of multiplying an $n\times n^{1-a}$ matrix and an $n^{1-a}\times n$ matrix.)
	
	Let $a=0.276724$, then $1-a=0.723276$. By convexity of the function $\omega(\cdot)$ \cite{LottiR83}, we have
	\[\omega(1-a)\le\frac{(a-0.25)\omega(0.7)+(0.3-a)\omega(0.75)}{0.75-0.7}.\]
	We substitute $\omega(0.7)\le 2.154399$ and $\omega(0.75)\le 2.187543$ \cite{GallU18}, and obtain:
	\[\omega(1-a)\le 20\cdot \left((a-0.25)\cdot 2.154399+(0.3-a)\cdot 2.187543\right)\le 2.169829.\]
	
	Therefore, given a directed graph $G=(V,E)$, there is a DSO with
	\[\tilde{O}\mleft(n^{\max\{2a+\omega(1-a),3-a\}}M\mright)=\tilde{O}(n^{2.723277}M)\]
	preprocessing time and $O(1)$ query time.
	
	\section{Proof of \cref{ob:fast-r-truncated-DSO-formal} and \cref{lemma:complexity-of-DSO}}\label{sec:proof-of-complexity}
	In this section, we prove \cref{ob:fast-r-truncated-DSO-formal}. Note that \cref{lemma:complexity-of-DSO} follows from \cref{ob:fast-r-truncated-DSO-formal} by setting $r=+\infty$.
	
	\ObsTruncatedBK*
	
	\subsection{The Preprocessing Algorithm}
	
	We review and slightly modify the preprocessing algorithm of \cite{BernsteinK09}. For convenience, we denote $\|p\|_r=\min\{\|p\|,r\}$ for any path $p$ and number $r$.
	
	\subparagraph{Assigning priorities.} We assign each vertex a \emph{priority}, which is independently sampled from the following distribution: for any positive integer $c$, each vertex has priority $c$ w.p.~$1/2^c$. Denote $c(v)$ the priority of the vertex $v$. With high probability, all of the following are true:\begin{itemize}
		\item The maximum priority is $O(\log n)$.
		\item For every $c\le O(\log n)$, there are $\tilde{O}(n/2^c)$ vertices with priority $c$.
		\item Let $C$ be a large enough constant. For every shortest path $uv$ with at least $C\cdot 2^c\log n$ edges, there is a vertex on $uv$ whose priority is greater than $c$.
	\end{itemize}
	
	In the following discussions, we will assume that all of the above assumptions hold.
	
	Fix a pair $u,v\in V$, let $s_i$ be the first vertex in $uv$ with priority $\ge i$, and $t_i$ be the last such vertex. Then we can write the path $uv$ as
	\[u\To s_1\To s_2\To\dots\To s_{O(\log n)} \To t_{O(\log n)} \To\dots \To t_1\To v.\]
	We say that the vertices $u,v,s_i,t_i$ are \emph{key vertices}, and the $i$-th key vertex is denoted as $k_i$. Then the path $uv$ can also be written as
	\[u=k_0\To k_1\To\dots\To k_{O(\log n)}=v.\]
	
	It is important to see that 
	\begin{equation}
	|k_ik_{i+1}|\le C\cdot 2^{\min\{c(k_i),c(k_{i+1})\}}\log n\label{eq:cover}
	\end{equation}
	for every valid $i$, as otherwise there will be another key vertex between $k_i$ and $k_{i+1}$.
	
	\subparagraph{Data structures for quick location.} Suppose we are given a query $(u,v,f)$, the first thing we should do is to ``locate'' $f$, i.e.~find the key vertices $k_i,k_{i+1}\in uv$ such that $f\in k_ik_{i+1}$. We will utilize the following data (see also~\cite{BernsteinK08}).
	\begin{itemize}
		\item $\CL[u,v,c]$ (for ``center left''): the first vertex in $uv$ with priority at least $c$;
		\item $\CR[u,v,c]$ (for ``center right''): the last vertex in $uv$ with priority at least $c$; and
		\item $\BCP[u,v]$ (for ``biggest center priority''): the maximum priority of any vertex on $uv$.
	\end{itemize}

	It is easy to compute these data in $\tilde{O}(n^2)$ time: for every $u\in V$, we perform a depth-first search on the outgoing shortest path tree $\Tout(u)$ to compute $\BCP[u, \cdot]$; for every $u\in V$ and each priority $c$, we also perform a depth-first search on $\Tout(u)$ to compute $\CL[u, \cdot, c]$ and $\CR[u, \cdot, c]$.

	In addition, for every $u,v\in V$, we store the key vertices on $uv$ into a hash table of size $O(\log n)$. Given a (vertex) failure $f$, we can output whether $f$ is among these key vertices on $uv$ in $O(1)$ worst-case time \cite{FredmanKS82}.
	
	\subparagraph{Data structures for avoiding a failure.} We use $\caD$ to preprocess the input graph. Then we compute the following data:\begin{enumerate}[({Data} a)]
		\item For every $u,v\in V$, and every $1\le i\le \min\{C\cdot 2^{c(u)}\log n, |uv|\}$, let $x_i$ be the $i$-th vertex in the path $uv$. (Here $u$ is the $0$-th vertex.) We compute and store the value $\|uv\diamond x_i\|_r$. Also, let $e_i$ be the edge from $x_{i-1}$ to $x_i$, we compute and store $\|uv\diamond e_i\|_r$. Symmetrically, let $x_{-i}$ be the \emph{last} $i$-th vertex in the path $vu$ (not $uv$!), and $e_{-i}$ be the edge from $x_{-i}$ to $x_{-(i-1)}$. For every $1\le i\le \min\{C\cdot 2^{c(u)}\log n, |vu|\}$, we compute and store $\|vu\diamond x_{-i}\|_r$ and $\|vu\diamond e_{-i}\|_r$. \label{item:avoid-priority}
		\item For every $u,v\in V$ and consecutive key vertices $k_i,k_{i+1}\in uv$ such that $k_i\ne u$ and $k_{i+1}\ne v$, let $y$ be the vertex in the portion $k_i\To k_{i+1}$ that maximizes $\|uv\diamond y\|_r$. We compute and store $\|uv\diamond y\|_r$.\label{item:avoid-bottleneck}
		\item For every $u,v\in V$ and key vertex $k_i\in uv$, we compute and store $\|uv\diamond k_i\|_r$.\label{item:avoid-key}
	\end{enumerate}
	
	For each priority $c\le\tilde{O}(1)$, there are $\tilde{O}(n/2^c)$ vertices $u$ whose priority is exactly $c$. In (Data \ref{item:avoid-priority}), we make $\tilde{O}(n2^c)$ queries for each such $u$ ($\tilde{O}(2^c)$ queries for each $v\in V$). Therefore in total, we make $\tilde{O}(n^2)$ queries in (Data \ref{item:avoid-priority}). We will show in \cref{sec:find-y} that we can compute (Data \ref{item:avoid-bottleneck}) using $\tilde{O}(n^2)$ queries to $\caD$ and $\tilde{O}(n^2)$ additional time. (Data \ref{item:avoid-key}) can be computed in $\tilde{O}(n^2)$ queries easily.

	\subsection{The Query Algorithm}\label{sec:BK09-query}
	Let $(u,v,f)$ be a query. We first check whether $f\in uv$ in the shortest path trees; if $f\not\in uv$, then it is easy to see that $\|uv\diamond f\|_r=\|uv\|_r$.
	
	If $f$ is a vertex failure, we check whether $f$ is a key vertex on $uv$ (that is, $f=k_i$ for some $i$), using the hash tables. If this is the case, we return $\|uv\diamond f\|_r$ stored in (Data \ref{item:avoid-key}) immediately.
	
	Otherwise, we start by finding two consecutive key vertices $k_i,k_{i+1}\in uv$ such that $f\in k_ik_{i+1}$. Recall that, if $\ell$ is the biggest priority of any vertex on $uv$, then the key vertices on $uv$ are
	\[(u=){\sf CL}[u,v,1]\To {\sf CL}[u,v,2]\To\dots\To {\sf CL}[u,v,\ell]\To{\sf CR}[u,v,\ell]\To\dots\To {\sf CR}[u,v,2]\To{\sf CR}[u,v,1](=v).\]
	
	Denote $a$ and $b$ as the ``tail'' and ``head'' of $f$ respectively. In particular, if $f$ is a vertex failure then $a=b=f$; if $f$ is an edge failure then it is an edge from $a$ to $b$. We can find $k_i$ in $O(1)$ time using the following procedure:
	\begin{itemize}
		\item If $\BCP[b, v] = \ell$, then $f$ is in the range $(u, \CR[u, v, \ell])$, so we have $k_i = \CL[u, v, \BCP[u, a]]$.
		\item Otherwise, $f$ is in the range $(\CR[u, v, \ell], v)$ and we can see that $k_i = \CR[u, v, \BCP[b, v] + 1]$.
	\end{itemize}

	We can find $k_{i+1}$ similarly. By Eq.~\eqref{eq:cover}, if $k_i=u$, then $|ub|\le C\cdot 2^{c(u)}\log n$, and we can look up the value $\|uv\diamond f\|_r$ from (Data \ref{item:avoid-priority}) directly. Similarly, if $k_{i+1}=v$ then we can also look up $\|uv\diamond f\|_r$ from (Data \ref{item:avoid-priority}).
	
	Now we assume that $k_i\ne u$ and $k_{i+1}\ne v$. A crucial observation is that
	\begin{equation}
	\|uv\diamond f\|=\min\{\|uk_{i+1}\diamond f\|+\|k_{i+1}v\|, \|uk_i\|+\|k_iv\diamond f\|,\|uv\diamond y\|\},
	\label{eq:triple-path}
	\end{equation}
	where $y$ is the vertex in $[k_i,k_{i+1}]$ that maximizes $\|uv\diamond y\|$. The proof of Eq.~\eqref{eq:triple-path} is as follows:\begin{enumerate}[(i)]
		\item If $uv\diamond f$ goes through $k_i$, then $\|uv\diamond f\|=\|uk_i\|+\|k_iv\diamond f\|$.
		\item If $uv\diamond f$ goes through $k_{i+1}$, then $\|uv\diamond f\|=\|uk_{i+1}\diamond f\|+\|k_{i+1}v\|$.
		\item If $uv\diamond f$ goes through neither $k_i$ nor $k_{i+1}$, then it avoids the entire portion of $k_i\To k_{i+1}$, thus also avoids $y$. We have $\|uv\diamond f\|\ge \|uv\diamond y\|$. But $\|uv\diamond y\|\ge \|uv\diamond a\|$ by definition of $y$, and $\|uv\diamond a\|\ge \|uv\diamond f\|$. (Recall that $a$ is the ``tail'' of $f$.) Thus $\|uv\diamond f\|=\|uv\diamond y\|$.\label{item:query-case-iii}
	\end{enumerate}
	
	It is easy to see that a similar equation holds for $r$-truncated DSOs:
	\begin{equation}
	\|uv\diamond f\|_r=\min\{\|uk_{i+1}\diamond f\|_r+\|k_{i+1}v\|, \|uk_i\|+\|k_iv\diamond f\|_r,\|uv\diamond y\|_r,r\},
	\label{eq:triple-path-truncated}
	\end{equation}
	where $y$ is any vertex in $[k_i,k_{i+1}]$ that maximizes $\|uv\diamond y\|_r$.
	
	Recall that we already know the values $\|uk_i\|$ and $\|k_{i+1}v\|$. To compute $\|uk_{i+1}\diamond f\|_r$, we note that if $f$ is the last $j$-th vertex or edge in $uk_{i+1}$, then $j\le C\cdot 2^{c(k_{i+1})}\log n$. Therefore we can look up the value of $\|uk_{i+1}\diamond f\|_r$ from (Data \ref{item:avoid-priority}). Similarly we can look up $\|k_iv\diamond f\|_r$. Finally, we can look up $\|uv\diamond y\|_r$ from (Data \ref{item:avoid-bottleneck}).
	
	We can see that the query time is $O(1)$.
	
	\subsection{Computing (Data \ref{item:avoid-bottleneck})}\label{sec:find-y}
	
	We will use the following notation. Let $p$ be a path from $u$ to $v$ which is fixed in context, and $a,b$ be two vertices in $p$. We will say that $a<b$ if $|p[u,a]|<|p[u,b]|$, i.e.~$a$ appears strictly before $b$ on the path $p$. Similarly, $a>b$, $a\le b$, $a\ge b$ mean $|p[u,a]|>|p[u,b]|$, $|p[u,a]|\le |p[u,b]|$, $|p[u,a]|\ge |p[u,b]|$ respectively.
	
	Let $u,v\in V$, and $s, t$ be two vertices on the path $uv$ such that $u < s < t < v$. Let $y\in [s,t]$ be the vertex in $[s,t]$ which maximizes $\|uv\diamond y\|_r$. We first show that \emph{assuming we have built some oracles}, we can find this vertex $y$ in $O(\log n)$ oracle calls and $O(\log n)$ additional time. The idea is to use a binary search described in \cite[Section 6]{BernsteinK09}.
	
	\begin{lemma}\label{lemma:find-y}
		Let $r$ be an integer, $u,v\in V$, and $s,t$ be two vertices on the path $uv$ such that $u<s<t<v$. Suppose we have the following oracles, each with $O(1)$ query time:
		\begin{itemize}
			\item an oracle that given a vertex $x\in st$, outputs $\|ut\diamond x\|_r$;
			\item an oracle that given a vertex $x\in st$, outputs $\|sv\diamond x\|_r$;
			\item an oracle that given an interval $[s',t']$ such that $s\le s'\le t'\le t$, outputs a vertex $x\in [s',t']$ that maximizes the value $\|ut\diamond x\|_r$.
		\end{itemize}
		Then we can find a vertex $y\in [s,t]$ which maximizes $\|uv\diamond y\|_r$ in $O(\log n)$ time.
	\end{lemma}
	\begin{proof}
		For any $y\in [s,t]$, we denote
		\[h(y)=\min\{\|ut\diamond y\|_r+\|tv\|, \|us\|+\|sv\diamond y\|_r,r\}.\]
		By Eq.~\eqref{eq:triple-path-truncated}, we have $\|uv\diamond y\|_r = \min\{h(y), \|uv\diamond y^\star\|_r\}$ where $y^\star$ is some vertex independent of $y$. Thus it suffices to find some $y\in [s,t]$ that maximizes $h(y)$.
		
		We use a binary search. Assume that we know the optimal $y$ is in some interval $[s',t']$, where $s\le s'<t'\le t$. (Initially we set $s'=s$ and $t'=t$.) If $|s't'|=O(1)$ then we can use brute force to find a vertex $y\in [s',t']$ that maximizes $h(y)$. Otherwise let $q$ be the middle point of $[s',t']$, and we use the third oracle to find a vertex $y\in [s',q]$ that maximizes $\|ut\diamond y\|_r$. There are two cases:\begin{itemize}
			\item If $\min\{\|ut\diamond y\|_r+\|tv\|,r\}=h(y)$, then we can restrict our attention to the interval $[q,t']$. This is because for every vertex $x\in [s',q]$,
			\[h(x)\le \min\{\|ut\diamond x\|_r+\|tv\|,r\}\le \min\{\|ut\diamond y\|_r+\|tv\|,r\} = h(y).\]
			\item Otherwise, $h(y)=\|us\|+\|sv\diamond y\|_r$. Since $\|ut\diamond y\|_r + \|tv\|$ is strictly larger than $h(y)$, we know that $uv\diamond y$ does not go through $t$. Therefore $sv\diamond y$ avoids every vertex in $[q,t']$. (See \cref{fig:binary-search}.) For every vertex $x\in [q,t']$,
			\[h(x)\le \|us\|+\|sv\diamond x\|_r \le \|us\|+\|sv\diamond y\|_r = h(y).\]
			It follows that we can restrict our attention to the interval $[s',q]$ now.
		\end{itemize}
		\begin{figure}[H]
			\centering
			\includegraphics[width=0.9\linewidth]{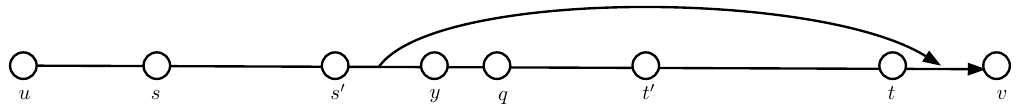}
			\caption{If $sv\diamond y$ does not go through $t$, then $sv\diamond y$ does not go through the whole interval $[q,t']$.}\label{fig:binary-search}
		\end{figure}
		Therefore, we can always shrink the length of our candidate interval $[s',t']$ by a half. It follows that we can find the desired vertex $y$ in $O(\log n)$ time.
	\end{proof}
	
	Now we show how to compute (Data \ref{item:avoid-bottleneck}) in $\tilde{O}(n^2)$ time (assuming that (Data \ref{item:avoid-priority}) is ready). The most crucial ingredient is the following Range Maximum Query (RMQ) structures (used in the third item of \cref{lemma:find-y}).
	
	For every $u,v\in V$, consider the following sequence (of length $\ell=\min\{|uv|-1, C\cdot 2^{c(v)}\log n\}$):
	\[\mleft(\|uv\diamond x_{-1}\|_r, \|uv\diamond x_{-2}\|_r, \dots, \|uv\diamond x_{-\ell}\|_r\mright),\]
	where $x_{-i}$ denotes the last $i$-th vertex in the path $uv$ ($v$ is the last $0$-th). We build an RMQ structure of this sequence, which given a query $(s,t)$ ($1\le s\le t\le \ell$), outputs a number $i\in[s,t]$ that maximizes $\|uv\diamond x_{-i}\|_r$. After we compute the above sequence, this data structure can be preprocessed in $O(\ell)$ time, and each query costs $O(1)$ time \cite{BenderF00}.
	
	For every priority $c\le O(\log n)$, there are $\tilde{O}(n/2^c)$ vertices $v$ of this priority, and for each vertex $v$ we construct $n$ RMQ structures (one for each $u\in V$) on length-$\tilde{O}(2^c)$ sequences. The total size of these RMQ structures is
	\[\sum_{c=1}^{O(\log n)}\tilde{O}(n/2^c)\cdot n\cdot \tilde{O}(2^c)=\tilde{O}(n^2).\]
	Therefore, these RMQ structures can be preprocessed in $\tilde{O}(n^2)$ time. (Note that every element $\|uv\diamond x_{-i}\|_r$ is already computed in (Data \ref{item:avoid-priority}).)
	
	To compute (Data \ref{item:avoid-bottleneck}), we enumerate $u,v,k_i,k_{i+1}$ where $k_i,k_{i+1}$ are consecutive key vertices in $uv$. There are $\tilde{O}(n^2)$ possible combinations of $(u,v,k_i,k_{i+1})$. As argued in \cref{sec:BK09-query}, we know that the following data are already computed in (Data \ref{item:avoid-priority}):
	\begin{itemize}
		\item $\|uk_{i+1}\diamond x\|_r$, for any $x\in k_ik_{i+1}$;
		\item $\|k_iv\diamond x\|_r$, for any $x\in k_ik_{i+1}$.
	\end{itemize}
	
	We also have the following RMQ oracles constructed above:
	\begin{itemize}
		\item An oracle that given any interval $[s',t']$ on the path $uv$ such that $k_i\le s'\le t'\le k_{i+1}$, finds the vertex $y\in [s',t']$ that maximizes $\|uk_{i+1}\diamond y\|_r$ in $O(1)$ time.
	\end{itemize}
	
	It follows from \cref{lemma:find-y} that we can find a vertex $y\in [k_i, k_{i+1}]$ that maximizes $\|uv\diamond y\|_r$ in $O(\log n)$ time. The total time for computing (Data \ref{item:avoid-bottleneck}) is thus $\tilde{O}(n^2)$.
	
	\section{Reporting the Actual Path}
	In this section, we modify our DSO so that it also supports \emph{path-reporting} queries: Given a query $(u, v, f)$, we want to report not only $\|uv\diamond f\|$, but also an \emph{actual} shortest path from $u$ to $v$ that avoids $f$. If the path has length $\ell$, then the query algorithm runs in $O(\ell)$ time. Unfortunately, the size of the new DSO becomes $\tilde{O}(n^{2+a})$ where $a$ is the constant fixed in \cref{sec:putting-it-together}, i.e.~$a = (3-\omega) / 2$ for undirected graphs and $a = 0.276724$ for directed graphs.\footnote{We remark that even if we do not need to report these paths, our preprocessing algorithm still needs space complexity $\tilde{O}(n^{2+a})$ (although the size of the DSO is $\tilde{O}(n^2)$).}
	
	Recall that our DSO is constructed as follows. Fix $r = n^a$ and $i_\star = \lceil\log_{3/2}(Mn/r)\rceil$. Let $\caD^0$ be an $r$-truncated DSO as in \cref{sec:r-small}, and $\caD^i = \Fast(\Extend(\caD^{i-1}))$ for every $1\le i\le i_\star$. Then $\caD^{i_\star}$ is a (normal) DSO.
	
	\subparagraph{Path-reporting structure for $\caD^0$.} Recall that the structure $\caD^0$ consists of subgraphs $G_i^\e$ and $G_i^\v$. For each subgraph, we can compute an implicit representation of the shortest paths of length at most $r$ as follows:\begin{itemize}
		\item If the graph is directed, we run the first $\lceil \log_{3/2}r\rceil$ iterations of the algorithm \textbf{rand-short-path} in \cite{Zwick02}. After that, the matrix $W$ in \cite[Figure 2]{Zwick02} encodes shortest paths of length at most $r$.
		\item If the graph is undirected, we simply use \cite[Section 4]{ShoshanZ99}.
	\end{itemize}
	
	Given the implicit representations, it is easy to report the actual path for any query of $\caD^0$. As we need to store $\tilde{O}(r)$ such representations, the size of our DSO becomes $\tilde{O}(n^2 r)$.
	
	\subparagraph{Path-reporting algorithm for $\caD^i$.} For every $1\le i\le i_\star$, the preprocessing algorithm of $\caD^i$ invokes $\tilde{O}(n^2)$ queries to $\Extend(\caD^{i-1})$. For each such query $(u_q, v_q, f_q)$:
	\begin{itemize}
		\item If $\|u_qv_q\diamond f_q\|\le r$, then the path $u_qv_q\diamond f_q$ can be retrieved from $\caD^0$, thus we do not need to store anything.
		\item Otherwise, suppose $i_q = \lfloor \log_{3/2}(\|u_qv_q\diamond f_q\| / r)\rfloor$, then $r(3/2)^{i_q} \le \|u_qv_q\diamond f_q\| < r(3/2)^{i_q+1}$. If $i_q \ge i$, then $\caD^i$ do not need the exact value of $\|u_qv_q\diamond f_q\|$; therefore we may assume $i_q < i$. This means that the query $(u_q, v_q, f_q)$ is captured by $\caD^{i_q + 1}$ but not by $\caD^{i_q}$. We store the ``hitting vertex'' $h_q$ in $\Extend(\caD^{i_q})$ that hits $\m(u_qv_q\diamond f_q)$ (as in \cref{sec:r-large}). Then, $u_qv_q\diamond f_q$ is the concatenation of $u_qh_q\diamond f_q$ and $h_qv_q\diamond f_q$, both of which can be retrieved from $\caD^{i_q}$.
	\end{itemize}
	
	Consider a query $(u, v, f)$. By the query algorithm in \cref{sec:BK09-query}, $uv\diamond f$ belongs to one of the following cases: \begin{enumerate}[(i)]
		\item the concatenation of $uk$ and $kv \diamond f$, for some key vertex $k$;
		\item the concatenation of $uk \diamond f$ and $kv$, for some key vertex $k$;
		\item $uv\diamond y$ for the vertex $y$ computed in \cref{sec:find-y}.
	\end{enumerate}
	In all of these cases, $uv\diamond f$ is the concatenation of a shortest path in $G$ (that is, $uk$, $kv$, or the empty path) and some $u_qv_q\diamond f_q$, where $(u_q, v_q, f_q)$ is a query recorded in the preprocessing algorithm.\begin{itemize}
		\item We can retrieve the shortest path in $G$ using the shortest path trees.
		\item If $\|u_qv_q\diamond f_q\| \le r$, we can retrieve it in $\caD^0$; otherwise we recursively find $u_qh_q\diamond f_q$ and $h_qv_q\diamond f_q$ in $\caD^{i_q}$, and concatenate them together to form $u_qv_q\diamond f_q$.
	\end{itemize}
	
	\subparagraph{Time complexity.} It remains to show that only $O(\ell)$ time is spent on retrieving a path of length $\ell$. The time complexity is proportional to $\ell$ plus the number of preprocessed queries $(u_q, v_q, f_q)$ that we access (over all DSOs $\caD^i$). Note that each $(u_q, v_q, f_q)$ corresponds to a hitting vertex $h_q$ on the reported path, and different queries correspond to different $h_q$'s. Therefore the number of such queries is at most $\ell$, which means the path-reporting algorithm only takes $O(\ell)$ time.

	\section{Open Questions}
	
	The main open problem after this work is to improve the preprocessing time for DSOs. Can we improve the preprocessing time for directed graphs to $\tilde{O}(n^{2.5286}M)$, matching the current best algorithm for APSP in directed graphs \cite{Zwick02}? A subsequent work by Yong Gu and the author \cite{GR21} improved the preprocessing time to $\tilde{O}(n^{2.5794}M)$, but there is still a gap between this time bound and the time bound for APSP.
	
	We can compute APSP for undirected graphs in $\tilde{O}(n^\omega M)$ time \cite{Seidel95, ShoshanZ99}, and another interesting question is whether there is a DSO for undirected graphs with the same preprocessing time (and constant query time).
	
	Finally, can we extend our technique to deal with negative weights? There are a few candidate definitions of ``$r$-truncated DSOs'' in this case, if we interpret $r$ as the number of edges in the shortest path, instead of the length of the shortest path. For example, we may define a DSO is $r$-truncated if on a query $(u, v, f)$, it outputs some value no less than $\|uv\diamond f\|$, and when $|uv\diamond f|\le r$, it outputs $\|uv\diamond f\|$ exactly. However, it seems that every definition of ``$r$-truncated DSO'' that we tried were not compatible with arguments in \cref{sec:find-y}.
	
	\bibliography{article}
	
	\appendix

	\section{Breaking Ties for Non-unique Shortest Paths}\label{sec:APBSP}
	In this section, we prove \cref{thm:property-APBSP} and \ref{thm:APBSP}.
	
	Recall that in a subgraph $G'$ of $G$, we denote $w_{G'}(u, v)$ as the largest (w.r.t.~$\prec$) edge $w$ such that there is a shortest path from $u$ to $v$ whose bottleneck edge (i.e.~smallest w.r.t.~$\prec$) is $w$. Also the path $\rho_{G'}(u, v)$ is defined, inductively from the smallest $\|uv\|$ to the largest $\|uv\|$, as follows: If $u=v$, then $\rho_{G'}(u, v)$ is the empty path; otherwise, suppose $w_{G'}(u, v)$ is an edge from $u^\star$ to $v^\star$, then $\rho_{G'}(u, v)$ is the concatenation of $\rho_{G'}(u, u^\star)$, $w_{G'}(u, v)$ and $\rho_{G'}(v^\star, v)$.
	
	\thmProperties*
	\begin{proof}
		In this proof we will always use the following notation: Let $w = w_{G'}(u, v)$, and suppose $w$ is an edge from $u^\star$ to $v^\star$. Denote $p_1 = \rho_{G'}(u, u^\star)$ and $p_2 = \rho_{G'}(v^\star, v)$. We will prove \cref{thm:property-APBSP} inductively, from the smallest $\|uv\|$ to the largest $\|uv\|$. \cref{thm:property-APBSP} is clearly true when $u=v$.
		
		\begin{figure}[H]
			\centering
			\includegraphics{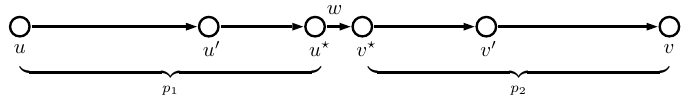}
		\end{figure}
		
		(Property \ref{prop:consistency}): If both $u', v' \in p_1$, then (Property \ref{prop:consistency}) is true by induction on $p_1$. Similarly, if both $u', v'\in p_2$, then (Property \ref{prop:consistency}) is also true by induction on $p_2$.
		
		Now suppose $u'\in p_1$ and $v'\in p_2$. Our first observation is that $w_{G'}(u', v') = w_{G'}(u, v) = w$. Actually, since $(\rho_{G'}(u, v))[u', v']$ is a shortest path from $u'$ to $v'$ with bottleneck edge $w$, we have $w_{G'}(u', v') \succeq w$. But if there is any shortest path $p'$ from $u'$ to $v'$ whose bottleneck edge is $\succ w$, the concatenation of $p_1[u, u']$, $p'$ and $p_2[v', v]$ will be a shortest path from $u$ to $v$ with bottleneck edge $\succ w$, a contradiction. Thus we can only have $w_{G'}(u', v') = w$.
		
		By induction on $p_1$, we have $\rho_{G'}(u', u^\star) = p_1[u', u^\star]$. By induction on $p_2$, we have $\rho_{G'}(v^\star, v) = p_2[v^\star, v]$. Therefore, $\rho_{G'}(u', v')$ is the concatenation of $p_1[u', u^\star]$, $w$, and $p_2[v^\star, v]$. This path coincides with the $u'\To v'$ portion of $\rho_{G'}(u, v)$.
		
		(Property \ref{prop:subgraph}): it suffices to prove $w_G(u, v) = w_{G'}(u, v)$, and use induction on $p_1$ and $p_2$. Since $G'$ is a subgraph of $G$, and the distance from $u$ to $v$ is the same in $G$ and $G'$, we have $w_G(u, v)\succeq w_{G'}(u, v)$. Since $\rho_G(u, v)$ is completely in $G'$ and the bottleneck edge of $\rho_G(u, v)$ is $w_G(u, v)$, we have $w_{G'}(u, v)\succeq w_G(u, v)$. This completes the proof.
	\end{proof}
	
	\thmAPBSP*
	\begin{proof}[Proof Sketch]
		The first step is to compute $w_{G'}(u, v)$ for every pair of vertices $u, v\in V$. This problem is exactly the same as the \emph{all-pairs bottleneck shortest paths} problem (APBSP; \cite[Theorem 4.4]{DuanP09}). However, Theorem 4.4 of \cite{DuanP09} only claims to work for unweighted graphs, so we verify that it also works for graphs with positive integer edge weights here.
		
		We start by using \cite{Zwick02} to compute $\|\rho_{G'}(u, v)\|$ for every $u, v\in V$. Then, for every $0\le i\le \lceil\log_{3/2}(nM)\rceil$ and $r=(3/2)^i$, assuming we have computed $w_{G'}(u, v)$ for every $u, v$ such that $\|\rho_{G'}(u, v)\|\le r$, we compute $w_{G'}(u, v)$ for every $u, v$ such that $\|\rho_{G'}(u, v)\|\le (3/2)r$. This is done as follows. We sample a vertex set $H$ by adding each vertex into $H$ independently w.p.~$\Theta(M\log n/r)$, then w.h.p.~$H$ hits $\m(\rho_{G'}(u, v))$ for every $u, v$ such that $r < \|\rho_{G'}(u, v)\| \le (3/2)r$. (Recall that $\m(\rho_{G'}(u, v))$ is the ``middle $1/3$ part'' of $\rho_{G'}(u, v)$.) Let $D, W$ be $V\times H$ matrices such that for every $u \in V, h \in H$,
		\begin{itemize}
			\item if $\|\rho_{G'}(u, h)\| \le r$, then $D[u, h] = \|\rho_{G'}(u, h)\|$ and $W[u, h] = w_{G'}(u, h)$;
			\item otherwise $D[u, h] = +\infty$ and $W[u, h] = -\infty$.
		\end{itemize}
		Then we compute the Distance-Max-Min product (Definition 4.3 of \cite{DuanP09}) of $(D, W)$ and $(D^{\sf T}, W^{\sf T})$ to obtain a $V\times V$ matrix $W'$ such that
		\[W'[u, v] = \max_{h: D[u, h] + D[h, v] = \|\rho_{G'}(u, v)\|}\mleft\{\min\mleft\{W[u, h], W[h, v]\mright\}\mright\}.\]
		We can see that $W'[u, v] = w_{G'}(u, v)$ for every $u, v\in V$ such that $r<\|\rho_{G'}(u, v)\|\le (3/2)r$. This completes the description of the APBSP algorithm.
		
		Now we analyze the time complexity. Let $|H|=n^s$, then $r=\tilde{O}(n^{1-s}M)$, and the finite entries in $D$ are in $\{1,2,\dots,r\}$. By \cite[Theorem 4.3]{DuanP09}, the Distance-Max-Min product takes
		\[O\mleft(\min\mleft\{n^{2+s}, r^{1/2}n^{1+s/2+\omega(s)/2}\mright\}\mright) = \tilde{O}\mleft(\min\mleft\{n^{2+s}, n^{(3+\omega(s))/2}M^{1/2}\mright\}\mright)\le \tilde{O}\mleft(n^{(3+\omega)/2}M^{1/2}\mright)\]
		time. Since we only execute $O(\log n)$ rounds of Distance-Max-Min product, the overall running time of the algorithm is $\tilde{O}\mleft(n^{(3+\omega)/2}M^{1/2}\mright)$.
		
		Now we construct the outgoing shortest path trees from the table of $w_{G'}(u, v)$ for every $u, v\in V$. It suffices to compute the parents of each node $v$ in the trees $\Tout(u)$ (which we denote as $\parent_u(v)$). We compute $\parent_u(v)$ inductively from the smallest $\|uv\|$ to the largest.
		
		Let $(u, v)$ be the pair of vertices we are processing, and assume that for every $u', v'\in V$ such that $\|u'v'\| < \|uv\|$, we have already computed $\parent_{u'}(v')$. Suppose $w_{G'}(u, v)$ is an edge from $u^\star$ to $v^\star$. If $v^\star=v$, then $\parent_u(v) = u^\star$. Otherwise, it is easy to see that $\|\rho_{G'}(v^\star, v)\| < \|\rho_{G'}(u, v)\|$ and $\parent_u(v) = \parent_{v^\star}(v)$. Thus we can compute every outgoing shortest path tree in $\tilde{O}(n^2)$ time. Similarly, the incoming shortest path trees can also be computed in $\tilde{O}(n^2)$ time.
	\end{proof}
	
\end{document}